\documentclass[journal]{IEEEtran}

\usepackage{amsmath,amssymb,amsthm,mathtools,bm}
\usepackage{graphicx}
\usepackage{booktabs}
\usepackage{enumitem}
\usepackage{cite}
\usepackage[hidelinks]{hyperref}
\usepackage{cleveref}
\usepackage{algorithm}
\usepackage{algorithmic}
\usepackage{siunitx}
\usepackage{subcaption}
\usepackage{multirow}
\usepackage{kotex}
\newtheorem{theorem}{Theorem}
\newtheorem{lemma}{Lemma}
\newtheorem{corollary}{Corollary}
\newtheorem{proposition}{Proposition}
\theoremstyle{remark}

\newcommand{\C}{\mathbb{C}}
\newcommand{\Rr}{\mathbb{R}}
\newcommand{\E}{\mathbb{E}}

\newcommand{\diag}{\operatorname{diag}}
\newcommand{\CN}{\mathcal{CN}}

\usepackage{kotex}

\title{Fundamental Limits of CSI Compression \\in FDD Massive MIMO}

\author{Bumsu Park, Youngmok Park, Chanho Park, and Namyoon Lee%
\thanks{B. Park, Y. Park, C. Park, and N. Lee are with the Department of Electrical Engineering, POSTECH, Pohang, South Korea
(e-mail: \{bumsupark, ympark1999, chanho26, nylee\}@postech.ac.kr).}}

\begin{document}
\maketitle

 \begin{abstract}
Channel state information (CSI) feedback in frequency-division duplex (FDD) massive multiple-input multiple-output (MIMO) systems is fundamentally limited by the high dimensionality of wideband channels. Classical transform coding (TC) provides an optimal solution when the CSI follows a single correlated Gaussian distribution. In practice, however, CSI datasets exhibit multiple propagation regimes arising from variations in user location, scattering geometry, and blockage. As a result, a single covariance model is often mismatched to the true channel statistics. In this paper, we model the stacked wideband CSI vector as a Gaussian-mixture source with a latent geometry state that represents different propagation environments. Each component corresponds to a locally stationary regime characterized by a correlated proper complex Gaussian distribution with its own covariance matrix. This representation captures the multimodal nature of practical CSI datasets while preserving the analytical tractability of Gaussian models. Motivated by this structure, we propose Gaussian-mixture transform coding (GMTC), a practical CSI feedback architecture that combines state inference with state-adaptive TC. The mixture parameters are learned offline from channel samples and stored as a shared statistical dictionary at both the user equipment (UE) and the base station. For each CSI realization, the UE identifies the most likely geometry state, encodes the corresponding label using a lossless source code, and compresses the CSI using the Karhunen–Loève transform matched to that state. We further characterize the fundamental limits of CSI compression under this model by deriving analytical converse and achievability bounds on the rate–distortion (RD) function. A key structural result is that the optimal bit allocation across all mixture components is governed by a single global reverse-waterfilling level. Simulations on the COST2100 dataset show that GMTC significantly improves the RD tradeoff relative to neural transform coding approaches while requiring substantially smaller model memory and lower inference complexity. These results indicate that near-optimal CSI compression can be achieved through state-adaptive TC without relying on large neural encoders.

\end{abstract}


\section{Introduction}
In frequency division duplex (FDD) massive multiple-input multiple-output (MIMO) orthogonal frequency division multiplexing (OFDM) systems, the user equipment (UE) estimates the downlink channel and feeds channel state information (CSI) to the base station (BS). With $N_t$ antennas and $N_c$ subcarriers, the raw CSI dimension scales as $N=N_tN_c$, so the feedback load grows with both array size and bandwidth \cite{Love2008_Overview,Jindal2006,Caire2006,Lee2014,LimitedFB2,GPIP2020}. This scaling makes CSI feedback a bottleneck in wideband massive MIMO.

Conditioned on a fixed propagation geometry---angles, delays, and large-scale path powers---the stacked CSI vector is well approximated as a correlated proper complex Gaussian vector whose covariance is determined by these slowly varying parameters \cite{COST2100,3GPP38901,NYUSIM,AdhikaryJSDM,Caire2025}. For a single covariance, classical transform coding (TC) based on the Karhunen--Lo\`eve transform (KLT) and reverse-waterfilling achieves the rate--distortion (RD) optimum \cite{Caire2025,goyal2002theoretical,cover2006elements,GPIP2020}. Practical CSI datasets such as COST2100, however, pool samples across multiple geometries (e.g., users, locations, visibility regions, blockage states). The resulting distribution is multi-modal, and any single covariance (hence any single transform) is inherently mismatched \cite{COST2100,3GPP38901}. We therefore model CSI as a Gaussian-mixture, where each component represents a locally stationary propagation regime.

Under experimentally validated and theoretically grounded Gaussian-mixture models for CSI, we ask two fundamental questions: (i) what is the information-theoretic limit of CSI compression, and (ii) what practical compression architecture can approach this limit with low complexity and flexible rate adaptivity? This paper addresses both questions.

 \subsection{Related Work}

  Classical CSI feedback methods quantize dominant beam directions or precoding subspaces using predetermined codebooks, including Grassmannian constructions and standardized beam codebooks \cite{love2003grassmannian,Love2008_Overview,Jindal2006}. These approaches are elegant and analytically tractable, but they scale poorly with dimension and become increasingly mismatched in wideband massive MIMO-OFDM settings where channel correlation varies jointly across space and frequency \cite{AdhikaryJSDM,COST2100,3GPP38901}. Compressive-sensing-based methods exploit angular--delay sparsity to reduce feedback overhead \cite{kuo2012compressive,CSbasedCSI2017}, while covariance- or subspace-based methods leverage structured long-term statistics and low-dimensional projections \cite{haghighatshoar2018massive,khalilsarai2018fdd,yang2023structured}. Their performance, however, can be brittle when sparsity or low-rank assumptions are violated, or when the channel statistics drift across propagation regimes.

More recently, deep-learning-based CSI feedback has cast CSI compression as a learned source coding problem; see \cite{Guo2022_Overview} for a comprehensive overview. Early autoencoder-based designs, most notably CsiNet, demonstrated that learned encoders and decoders can outperform classical codebook feedback on benchmark datasets such as COST2100 \cite{deeplearning2018,COST2100}, and subsequent work improved both performance and practicality using lightweight convolutional architectures, attention mechanisms, and learned optimization strategies \cite{Lightweight_CSI_feedback_2021,Attention_CSI_Feedback_2022,CNN2020,Learningtopitmize_2025}. In parallel, the CSI feedback literature has increasingly adopted ideas from learned image compression by treating the complex CSI matrix as a two-channel image and applying a learned analysis transform, quantization module, probabilistic entropy model, and synthesis transform, i.e., the nonlinear transform coding (NTC) viewpoint \cite{balle2020nonlinear}. Building on this perspective, recent CSI compressors incorporate variable-rate and entropy-aware designs \cite{park2024multi,park2025transformer,Liotopoulos2025OFSQ}; for example, \cite{park2025transformer} combines a transformer-based encoder/decoder with a learned entropy model and multi-level scalar quantization, achieving strong RD performance on COST2100 while highlighting the associated resource tradeoffs. Beyond deterministic decoders, generative diffusion models have also emerged as powerful reconstruction engines in neural compression \cite{Theis2022GaussianDiffusion,Yang2024ConditionalDiffusionCompression}. In particular, \cite{Kim2025DiffusionCSI} proposes diffusion-model-based compression of MIMO CSI with side information, using fixed-rate encoding via a trainable codebook and a conditional reverse-diffusion decoder, while residual diffusion architectures have been developed for variable-rate CSI compression and joint source--channel coding \cite{Ankireddy2025ResidualDiffusion}. Although these generative approaches can deliver impressive RD gains, their iterative sampling procedures typically incur substantially higher decoder complexity than TC-based schemes.

At a more fundamental level, several recent works have used RD theory to study the limits of CSI feedback systems \cite{Caire2,Kim2025,Caire2025}. For example, the asymptotic behavior of the overall feedback distortion versus the downlink signal-to-noise ratio under finite training was characterized in \cite{Caire2}, while practical schemes based on entropy-coded scalar quantization and truncated Karhunen--Lo\`eve feedback were developed in \cite{Kim2025,Caire2025}. However, these analyses are built around a \emph{single} correlated Gaussian CSI model. In realistic wideband systems, CSI datasets are often \emph{multi-modal} because the underlying propagation geometry varies across users, locations, visibility regions, and blockage states. To the best of our knowledge, the \emph{RD characterization of CSI compression for such multi-modal channel statistics has not been developed before}. In this paper, we adopt a model-based perspective by representing a multi-modal massive MIMO channel dataset through a Gaussian-mixture. This approach not only facilitates the development of a user-friendly lossless-lossy hybrid feedback architecture but also provides, more importantly, the first tractable information-theoretic RD characterization for CSI compression for multi-modal CSI statistics.

\subsection{Contributions}
We summarize the main contributions.

\begin{itemize}[leftmargin=*,itemsep=2pt]



\item \textbf{Universal CSI data model:}
We present a comprehensive statistical model for CSI datasets, specifically representing the stacked wideband CSI vector as a Gaussian-mixture with an associated latent geometry label. Each component of this mixture captures a locally stationary propagation regime, characterized by a correlated proper complex Gaussian distribution, complete with its own covariance structure. Notably, the proposed CSI distribution is universal, enabling accurate modeling of any arbitrary multimodal source distribution, in accordance with the universal approximation theorem \cite{Bishop2006PRML}. Furthermore, Gaussian-mixtures are extensively employed in high-dimensional statistics, serving not only as universal approximators but also because they facilitating mathematically tractable problem-solving.

\item \textbf{Gaussian-Mixture transform coding for lossless-lossy hybrid CSI compression:}
 Building on the differential entropy structure of the Gaussian-mixture, we introduce Gaussian-mixture transform coding (GMTC), a UE-friendly CSI feedback architecture characterized by low complexity encoder and decoder along with flexible rate adaptivity. The mixture parameters are learned offline from channel samples and stored as a shared covariance dictionary at both the UE and the BS. For each CSI realization, the UE conducts a maximum a posteriori (MAP) geometry-state selection and entropy-codes the state label in a lossless manner. Once the state is selected, the UE applies lossy compression through a component-matched KLT on the instantaneous CSI vector. The transform coefficients are then quantized with a rate allocation determined by a single global reverse-waterfilling level, and the entropy-coded coefficient indices are fed back to the BS. The BS subsequently reconstructs the CSI by decoding the label and coefficients, followed by applying the inverse KLT using the shared dictionary. This process enables both rate-adaptive CSI feedback and accurate CSI estimation, enhancing the overall CSI feedback performance of the system.

\item \textbf{Information-theoretical RD bounds:}
 To characterize the fundamental limit of CSI compression under clustered channel statistics, we derive an analytical converse-achievability characterization for Gaussian-mixture CSI compression. This characterization encompasses a genie-aided lower bound (where the state is revealed at no cost) and a label-aware achievable upper bound that involves explicit label transmission. A significant structural outcome of this analysis is the identification of a global reverse-waterfilling rule, wherein a single water level governs the optimal rate or distortion allocation across all mixture components and their respective eigenmodes. The gap between the upper and lower bounds is confined to a maximum of the amortized state-label entropy. Moreover, this gap can diminish as the number of antennas or the channel coherence time increases.

\item \textbf{Empirical validation:} Simulations on synthetic Gaussian-mixture CSI and on COST2100 quantify the mismatch penalty of single covariance TC and show that GMTC approaches the oracle-label benchmark. At comparable modeling budgets, GMTC significantly improves the RD tradeoff relative to representative learned NTC coders while retaining TC-like encoder complexity. In the COST2100 experiments, GMTC achieves roughly a $10$ dB improvement in reconstruction quality at similar rates, with about a $4\times$ smaller memory footprint and approximately $1650\times$ lower inference complexity than shifted-window vision transformer (Swin-ViT) \cite{liu2021swin} based NTC codecs. These results indicate that approaching the near-optimal RD performance does not require heavy neural inference at the encoder; most of the gains can be obtained by state-adaptive TC with the reverse-waterfilling rate allocation optimized for the multi-modal CSI distribution. 
\end{itemize}

\section{Gaussian-Mixture CSI Model}
\label{sec:csi_channel_model}

This section introduces the statistical CSI model used throughout the paper. Our starting point is the widely used and measurement-supported observation that \emph{conditioned on a fixed local propagation geometry}, the stacked wideband CSI vector is well approximated as a correlated proper complex Gaussian random vector \cite{COST2100,3GPP38901,AdhikaryJSDM,NYUSIM,Han2024,Kimjy2025}. The corresponding correlation structure is induced by physical scattering through the angles, delays, and powers of dominant multipath components. We then lift this \emph{local} Gaussian model to a \emph{global} Gaussian-mixture model to capture the multi-modality that appears in practical CSI datasets when the underlying propagation geometry changes across users, locations, and blockage states.

\subsection{Geometry-Induced Correlated Gaussian CSI}
\label{subsec:geom_gaussian_channel}

We consider a standard downlink massive MIMO--OFDM model \cite{Caire2025}. The BS employs a uniform linear array (ULA) with $N_t$ half-wavelength-spaced antennas and transmits over $N_c$ OFDM subcarriers. Let
\begin{equation}
\mathbf H_k\in\C^{N_c\times N_t}
\end{equation}
denote the frequency-domain downlink channel matrix of user $k$, where the $n$th row corresponds to the $n$th subcarrier and the $m$th column corresponds to the $m$th BS antenna. To obtain a physically consistent covariance model with joint spatial--frequency correlation, we adopt a multipath representation with $L$ dominant paths \cite{Caire2025}:
\begin{equation}
\label{eq:instant_channel_matrix}
\mathbf H_k
=
\sum_{\ell=1}^{L}
g_{k,\ell}\,
\mathbf b(\tau_{k,\ell})\,
\mathbf a(\theta_{k,\ell})^{\mathsf T},
\end{equation}
where $g_{k,\ell}\sim\CN(0,\gamma_{k,\ell})$ is the complex gain of the $\ell$-th path, $\theta_{k,\ell}$ is its angle, and $\tau_{k,\ell}$ is its delay. The spatial steering vector of the $N_t$-antenna ULA is $\mathbf a(\theta)\in\C^{N_t}$ with entries
\begin{equation}
\label{eq:spatial_steering}
[\mathbf a(\theta)]_m
=
e^{j\pi(m-1)\sin(\theta)},
\qquad
m=1,\ldots,N_t,
\end{equation}
and the frequency-response vector across subcarriers is $\mathbf b(\tau)\in\C^{N_c}$ with entries
\begin{equation}
\label{eq:delay_vector}
[\mathbf b(\tau)]_n
=
e^{-j2\pi(n-1)\Delta_f \tau},
\qquad
n=1,\ldots,N_c,
\end{equation}
where $\Delta_f$ is the subcarrier spacing.

Vectorizing \eqref{eq:instant_channel_matrix} yields the stacked CSI vector
\begin{equation}
\label{eq:instant_channel_vec}
\mathbf h_k
=
\mathrm{vec}(\mathbf H_k)
=
\sum_{\ell=1}^{L}
g_{k,\ell}\,
\big(\mathbf a(\theta_{k,\ell})\otimes \mathbf b(\tau_{k,\ell})\big),
\end{equation}
where $\mathbf h_k\in\C^N$ and $N\triangleq N_tN_c$.

Define the (slowly varying) local geometry parameter
\begin{equation}
\mathcal G_k
\triangleq
\{(\theta_{k,\ell},\tau_{k,\ell},\gamma_{k,\ell})\}_{\ell=1}^{L}.
\end{equation}
Conditioned on $\mathcal G_k$, the path gains $\{g_{k,\ell}\}$ are jointly proper complex Gaussian and enter linearly in \eqref{eq:instant_channel_vec}. Therefore, the stacked CSI is itself proper complex Gaussian:
\begin{equation}
\label{eq:local_gaussian}
\mathbf h_k \mid \mathcal G_k
\sim
\CN(\mathbf 0,\mathbf R(\mathcal G_k)),
\end{equation}
with geometry-induced covariance
\begin{equation}
\label{eq:multipath_covariance}
\mathbf R(\mathcal G_k)
=
\sum_{\ell=1}^{L}
\gamma_{k,\ell}\,
\Big(
\mathbf a(\theta_{k,\ell})\mathbf a(\theta_{k,\ell})^{\mathsf H}
\Big)
\otimes
\Big(
\mathbf b(\tau_{k,\ell})\mathbf b(\tau_{k,\ell})^{\mathsf H}
\Big).
\end{equation}
The correlation matrix in \eqref{eq:multipath_covariance} defines the role of geometry explicitly: angles govern the spatial correlation, delays govern the frequency correlation, and path powers determine how energy is distributed across the joint spatial--frequency modes.

\vspace{0.1cm}
\noindent{\bf Remark 1 (Local Gaussianity versus global nonstationarity):}
The model \eqref{eq:local_gaussian} is accurate \emph{locally}, i.e., for a fixed geometry or over a short interval where the scattering configuration is effectively constant. In a realistic CSI dataset, however, $\mathcal G_k$ varies across users and time due to mobility, sector transitions, visibility-region changes, and blockage. Consequently, the unconditional CSI distribution obtained by pooling many samples is generally \emph{non-Gaussian} and often \emph{multi-modal}, so a single covariance matrix is too restrictive.

\subsection{Gaussian-Mixture Model across Propagation Geometries}
\label{subsec:gmm_geometry}

To capture geometry variation across CSI samples, we introduce a latent \emph{geometry state} random variable
\begin{equation}
C_k \in \{1,\dots,K\},
\end{equation}
where each state represents a distinct propagation regime (e.g., a visibility region or a characteristic angle--delay cluster configuration). Conditioned on $C_k=c$, we model the stacked CSI as
\begin{equation}
\label{eq:conditional_gmm_channel}
\mathbf h_k \mid (C_k=c)
\sim
\CN(\mathbf 0,\mathbf R_c),
\end{equation}
where $\mathbf R_c$ is the covariance associated with state $c$. A natural parametric form for $\mathbf R_c$ inherits the same geometry-driven structure as \eqref{eq:multipath_covariance}:
\begin{equation}
\label{eq:component_covariance}
\mathbf R_c
=
\sum_{\ell=1}^{L_c}
\gamma_{c,\ell}\,
\Big(
\mathbf a(\theta_{c,\ell})\mathbf a(\theta_{c,\ell})^{\mathsf H}
\Big)
\otimes
\Big(
\mathbf b(\tau_{c,\ell})\mathbf b(\tau_{c,\ell})^{\mathsf H}
\Big),
\end{equation}
where $\{(\theta_{c,\ell},\tau_{c,\ell},\gamma_{c,\ell})\}_{\ell=1}^{L_c}$ describe the representative multipath geometry of state $c$.

Let $\pi_c=\Pr(C_k=c)$ with $\pi_c>0$ and $\sum_{c=1}^{K}\pi_c=1$. The unconditional CSI distribution becomes the proper complex Gaussian-mixture
\begin{equation}
\label{eq:gmm_csi}
p(\mathbf h)
=
\sum_{c=1}^{K}
\pi_c\,
\CN(\mathbf h;\mathbf 0,\mathbf R_c).
\end{equation}

\vspace{0.1cm}
\noindent{\bf Remark 2 (Gaussian-mixtures  can approximate real CSI datasets accurately):}
A practical CSI dataset is typically formed by pooling channel realizations across many users, locations, and time instants. Even when each realization is locally well modeled by \eqref{eq:local_gaussian}, the underlying geometry $\mathcal G$ varies continuously across the dataset. Formally, this induces an (infinite) mixture of Gaussians of the form
\begin{align}
    p(\mathbf h)=\int \CN\!\big(\mathbf h;\mathbf 0,\mathbf R(\mathcal G)\big)\,p_{\mathcal G}(\mathcal G)\,d\mathcal G. \label{eq:mixture}
\end{align}
A finite Gaussian-mixture \eqref{eq:gmm_csi} can be viewed as a tractable approximation to this integral mixture obtained by quantizing the geometry space into $K$ representative regimes (e.g., via clustering/EM on channel samples). This is particularly well aligned with geometry-based stochastic channel models and measured datasets such as COST2100 and 3GPP-type scenarios, where CSI statistics change across visibility regions and cluster configurations \cite{COST2100,3GPP38901}. In addition, Gaussian-mixtures are universal density models: with sufficiently large $K$, they can approximate a broad class of continuous distributions, while retaining conditional Gaussian structure that enables analytic RD analysis and low-complexity TC.


\subsection{Universality and Entropy of Gaussian-Mixtures}
\label{subsec:gmm_prelim}

Although the CSI vectors in this paper are complex valued, it is convenient to state several statistical facts in an equivalent real-valued form. A Gaussian-mixture model with $K$ components has density
\begin{equation}
\label{eq:gmm_model_method}
p_K(\mathbf x)
\triangleq
\sum_{k=1}^{K}\pi_k\,
\mathcal{N}(\mathbf x;\boldsymbol{\mu}_k,\boldsymbol{\Sigma}_k),
\qquad \mathbf x\in\Rr^{N},
\end{equation}
where $\pi_k\ge 0$, $\sum_{k=1}^{K}\pi_k=1$, $\boldsymbol{\mu}_k\in\Rr^{N}$ is the mean vector, and $\boldsymbol{\Sigma}_k \succeq \mathbf 0$ is the covariance matrix of the $k$th component.

We next recall two fundamental properties that motivate Gaussian-mixtures as a CSI source model.

\begin{lemma}[Universality of Gaussian-mixtures]
\label{lemma1}
Let $p(\mathbf x)$ be a continuous and integrable probability density on $\Rr^{N}$. For any $\varepsilon>0$, there exist an integer $K$, weights $\{\pi_k\}_{k=1}^{K}$ with $\pi_k\ge 0$ and $\sum_{k=1}^{K}\pi_k=1$, and Gaussian parameters $\{(\boldsymbol{\mu}_k,\boldsymbol{\Sigma}_k)\}_{k=1}^{K}$ such that the Gaussian-mixture in \eqref{eq:gmm_model_method} satisfies
\begin{equation}
\|p(\mathbf x)-p_K(\mathbf x)\|_1<\varepsilon.
\end{equation}
\end{lemma}

\begin{proof} The formal proof is provided in \cite{GMproof, Bishop2006PRML}. For the sake of completeness, we will provide a sketch of the proof here. 

Let $d=N_r$ and let $\varphi_\sigma(\mathbf{x})\triangleq\mathcal{N}(\mathbf{x};\mathbf{0},\sigma^2\mathbf{I}_d)$.
Define the Gaussian-smoothed density
\[
p_\sigma(\mathbf{x}) \triangleq (p*\varphi_\sigma)(\mathbf{x})
= \int_{\mathbb{R}^d} p(\mathbf{y})\,\mathcal{N}(\mathbf{x};\mathbf{y},\sigma^2\mathbf{I}_d)\,d\mathbf{y}.
\]
Since $\{\varphi_\sigma\}$ is an approximate identity as $\sigma\rightarrow 0$and $p\in L^1(\mathbb{R}^d)$ is continuous,
we have $\|p-p_\sigma\|_1\to 0$ as $\sigma\rightarrow 0$; hence choose $\sigma$ such that
$\|p-p_\sigma\|_1<\varepsilon/2$.

Next, $p_\sigma$ is an \emph{infinite} Gaussian mixture with mixing density $p(\mathbf{y})$.
Choose $R>0$ so that $\int_{\|\mathbf{y}\|>R}p(\mathbf{y})\,d\mathbf{y}$ is sufficiently small, and partition the ball
$\{\|\mathbf{y}\|\le R\}$ into finitely many cells $\{C_k\}_{k=1}^K$ of small diameter.
Pick representatives $\boldsymbol{\mu}_k\in C_k$ and weights
$\pi_k\triangleq\int_{C_k}p(\mathbf{y})\,d\mathbf{y}$.
Define the finite Gaussian mixture as in \eqref{eq:gmm_model_method}
and (if needed) add one extra component to absorb the tail mass so that $\sum_{k=1}^K\pi_k=1$.

By continuity of translations in $L^1$ for the Gaussian kernel, making the cell diameters small ensures that replacing
$\mathcal{N}(\mathbf{x};\mathbf{y},\sigma^2\mathbf{I}_d)$ by
$\mathcal{N}(\mathbf{x};\boldsymbol{\mu}_k,\sigma^2\mathbf{I}_d)$ for $\mathbf{y}\in C_k$ yields
$\|p_\sigma-p_K\|_1<\varepsilon/2$ (and the tail contributes arbitrarily little by the choice of $R$).
Finally, the triangle inequality gives
\[
\|p-p_K\|_1 \le \|p-p_\sigma\|_1+\|p_\sigma-p_K\|_1 < \varepsilon,
\]which completes the proof.
\end{proof}

\noindent{\bf Remark 3 (Why universality matters for real CSI datasets):}
Lemma~\ref{lemma1} provides a formal justification for using \eqref{eq:gmm_model_method} to fit practical CSI distributions. Even if the true dataset distribution is induced by a complicated geometry process (cf. \eqref{eq:mixture}) and is not exactly a finite mixture, a sufficiently rich Gaussian-mixture can approximate it arbitrarily well in total variation. Crucially, the mixture retains \emph{conditional Gaussian structure} (one covariance per state), which enables tractable RD analysis and low-complexity TC once the active state is known or inferred.

\begin{lemma}[Differential entropy of Gaussian-mixtures]
\label{lemma2}
Let $C\in\{1,\dots,K\}$ be a discrete random variable with $\Pr(C=k)=\pi_k$, and let
$ 
\mathbf x\mid\{C=k\}\sim \mathcal N(\boldsymbol{\mu}_k,\boldsymbol{\Sigma}_k),
$ for $k\in [K]$.
Then
\begin{equation}
\label{eq:gmm_entropy_identity}
h(\mathbf x)=h(\mathbf x\mid C)+I(\mathbf x;C),
\end{equation}
where $h(\mathbf x)$ is the differential entropy under $p_K(\mathbf x)$ in \eqref{eq:gmm_model_method} and $I(\mathbf x;C)$ is the mutual information between $\mathbf x$ and $C$. Consequently,
\begin{equation}
\label{eq:gmm_entropy_bounds}
h(\mathbf x\mid C)\le h(\mathbf x)\le h(\mathbf x\mid C)+H(C),
\end{equation}
where $h(\mathbf x\mid C)=
\frac{1}{2}\sum_{k=1}^{K}\pi_k
\log_2\!\big((2\pi e)^{N}\det(\boldsymbol{\Sigma}_k)\big)$ and $H(C)=
-\sum_{k=1}^{K}\pi_k\log_2 \pi_k $
\end{lemma}

\begin{proof}
Although the detailed proof is provided in \cite{lee2026}, we present a sketch of the proof here for the completeness of the paper.
 
Using $p(\mathbf{x},k)=\pi_k\,p(\mathbf{x}\mid k)$, the chain rule for the mixed discrete--continuous pair $(\mathbf{x},C)$ gives
$ 
h(\mathbf{x},C)=H(C)+h(\mathbf{x}\mid C)$ and $
h(\mathbf{x},C)=h(\mathbf{x})+H(C\mid \mathbf{x}).$ Equating the two expressions yields
\begin{align}
   h(\mathbf{x})&=h(\mathbf{x}\mid C)+H(C)-H(C\mid \mathbf{x})\nonumber\\
&= h(\mathbf{x}\mid C)+I(\mathbf{x};C), 
\end{align}
which proves \eqref{eq:gmm_entropy_identity}. Since $0\le I(\mathbf{x};C)\le H(C)$, we obtain
\[
h(\mathbf{x}\mid C)\le h(\mathbf{x})\le h(\mathbf{x}\mid C)+H(C),
\]
which completes the proof. 
\end{proof}
 
\noindent{\bf Remark 4 (Entropy decomposition for Gaussian-mixtures):}
As proven in \cite{lee2026}, Lemma~\ref{lemma2} shows that the descriptive complexity of a mixture source separates into two parts: (i) the average complexity of the \emph{local} Gaussian components, $h(\mathbf x\mid C)$, and (ii) the discrete uncertainty of the latent state, captured (up to at most $H(C)$ bits) by the label $C$. This decomposition provides a high-level idea of the structure exploited by GMTC: CSI feedback can be organized as a \emph{lossless} label stream (state signaling) together with \emph{lossy} component-matched TC of the Gaussian coefficients within the selected state.

\section{CSI Compression Problem}
\label{sec:rd_formulation}

We now cast CSI compression as a \emph{lossy source coding} problem. The goal is to represent the downlink CSI at the BS using a limited number of feedback bits while controlling the reconstruction distortion. This formulation provides the information-theoretic limits of CSI compression from a RD perspective.

\subsection{Source Model and Feedback Code}
\label{subsec:encoder_decoder}

At each CSI feedback interval $t$, the UE observes (or estimates) the stacked wideband CSI vector
\begin{equation}
\mathbf h_t \in \C^{N}, 
\end{equation}
where $\mathbf h_t=\mathrm{vec}(\mathbf H_t)$ is obtained by vectorizing the $N=N_c\times N_t$ frequency-domain channel matrix. We model $\{\mathbf h_t\}_{t\ge 1}$ as samples from an underlying source distribution $p(\mathbf h)$ (specified in \cref{sec:csi_channel_model}). For the information-theoretic characterization, we adopt the standard memoryless model, which draws independent and identically distributed (i.i.d.) $n$-blocklength CSI samples from $ p(\mathbf h)$, i.e., $\mathbf h_1,\ldots,\mathbf h_n$, and denote $\mathbf h^n=(\mathbf h_1,\ldots,\mathbf h_n)$.

A \emph{rate-$R$ CSI feedback code} of blocklength $n$ consists of:
\begin{itemize}[leftmargin=*, itemsep=2pt]
\item an encoder
\begin{equation}
f_n:\;(\C^{N})^n \rightarrow \{1,\dots,2^{nNR}\},
\end{equation}
\item a decoder
\begin{equation}
g_n:\;\{1,\dots,2^{nNR}\}\rightarrow (\C^{N})^n.
\end{equation}
\end{itemize}
The encoder maps the CSI sequence $\mathbf h^n$ to a feedback index
\begin{equation}
s_n=f_n(\mathbf h^n),
\end{equation}
and the decoder reconstructs
\begin{equation}
\hat{\mathbf h}^n
=
g_n(s_n)
=
(\hat{\mathbf h}_1,\ldots,\hat{\mathbf h}_n).
\end{equation}
The rate $R$ is measured in \emph{bits per complex channel dimension}, so the total number of feedback bits per CSI vector is $NR$.

\begin{figure*}[t]
    \centering
    \includegraphics[width=1\textwidth]{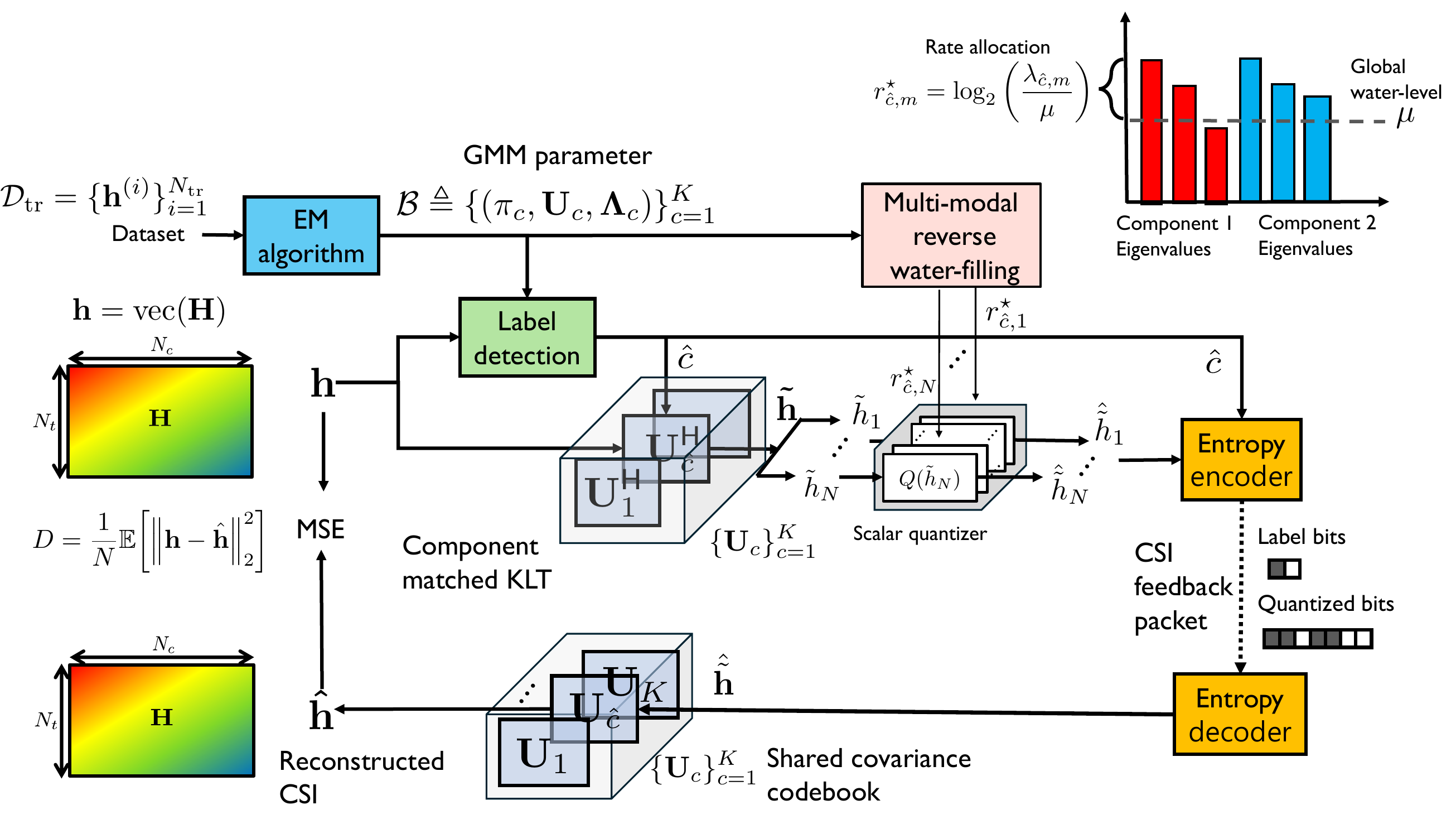}
    \caption{The proposed GMTC codec architecture for CSI compression, where $N_{\mathrm{tre}}$ denotes the size of the training set.}
    \label{Fig:1}
 \end{figure*}

\subsection{Distortion Measure}
\label{subsec:distortion_measure}

We measure CSI reconstruction quality by the mean-squared error (MSE) per complex channel dimension:
\begin{equation}
\label{eq:csi_block_distortion}
D
\triangleq
\mathbb{E}\!\left[
\frac{1}{nN}
\sum_{t=1}^{n}
\left\|
\mathbf h_t-\hat{\mathbf h}_t
\right\|_2^2
\right].
\end{equation}
Under the memoryless model, the distortion constraint is single-letter:
\begin{equation}
\label{eq:csi_single_letter_distortion}
D
=
\frac{1}{N}\mathbb{E}\!\left[\left\|\mathbf h-\hat{\mathbf h}\right\|_2^2\right].
\end{equation}

In simulations, performance is often reported using normalized MSE (NMSE), denoted as $\mathrm{NMSE}
=
\frac{\mathbb{E}\!\left[\|\mathbf h-\hat{\mathbf h}\|_2^2\right]}
{\mathbb{E}\!\left[\|\mathbf h\|_2^2\right]}.$ Since $\mathbb{E}\!\left[\|\mathbf h\|_2^2\right]/N$ is a fixed source-power normalization, minimizing MSE or NMSE leads to the same optimal coding rule for a given source model.

\subsection{Achievability and RD Function}
\label{subsec:csi_rd_function}

A pair $(R,D)$ is \emph{achievable} if there exists a sequence of $(n,R)$ CSI feedback codes such that the expected distortion in \eqref{eq:csi_block_distortion} is at most $D$ (as $n\to\infty$). The \emph{RD function} $R^\star(D)$ specifies the minimum rate required to achieve distortion $D$.

For a memoryless source with squared-error distortion, Shannon's RD theorem gives the single-letter characterization \cite{cover2006elements,BergerRD}:
\begin{equation}
\label{eq:csi_rd_function}
R^\star(D)
=
\frac{1}{N}
\inf_{p(\hat{\mathbf h}|\mathbf h):
\ \frac{1}{N}\mathbb{E}\|\mathbf h-\hat{\mathbf h}\|_2^2\le D}
I(\mathbf h;\hat{\mathbf h}),
\end{equation}
where $I(\mathbf h;\hat{\mathbf h})$ is mutual information and $R^\star(D)$ is measured in bits per complex channel dimension. Equivalently, one may define the distortion--rate function
\begin{equation}
D^\star(R)
=
\inf_{\substack{p(\hat{\mathbf h}|\mathbf h):\\
\frac{1}{N}I(\mathbf h;\hat{\mathbf h})\le R}}
\frac{1}{N}\mathbb{E}\!\left[\|\mathbf h-\hat{\mathbf h}\|_2^2\right].
\end{equation}

Although \eqref{eq:csi_rd_function} is conceptually clean, it is rarely tractable for high-dimensional, non-Gaussian CSI distributions. Closed-form solutions are known for special classes such as correlated Gaussian sources (via KLT and reverse-waterfilling), but realistic CSI datasets are multi-modal and nonstationary. This motivates the Gaussian-mixture source model in \cref{sec:csi_channel_model}, the GMTC architecture developed next, and the RD bounds analysis provided in subsequent sections.

\section{Gaussian-Mixture Transform Coding}
\label{sec:gmtc}

We now put forth GMTC, a CSI feedback architecture tailored to the Gaussian-mixture model in \cref{sec:csi_channel_model} and the source coding formulation in \cref{sec:rd_formulation}. 
GMTC is a lossless--lossy hybrid code. The overall encoding and decoding architecture is illustrated in Fig. \ref{Fig:1}. Specifically, each CSI realization is assigned to a geometry state; the UE sends the state label using lossless coding and then compresses the CSI using TC matched to the corresponding covariance. The required transforms and eigenvalue spectra are learned offline and stored as a shared dictionary at the UE and BS, so the online payload consists only of the label and the coded transform coefficients. In \cref{sec:fundamental}, we show that label-aware GMTC attains the conditional Gaussian benchmark up to a normalized label-entropy term, and that the optimal allocation is governed by a single multi-modal reverse-waterfilling level.

\subsection{Offline Learning and Shared Transform Dictionary}
\label{subsec:offline_dictionary}

GMTC separates the problem into two time scales. Offline, the BS learns a small set of \emph{covariance states} that summarize the long-term multi-modality of the CSI dataset; online, the UE only needs to identify the active state and apply the corresponding transform.

Let $\mathcal D_{\rm tr}=\{\mathbf h^{(i)}\}_{i=1}^{N_{\rm tr}}$ be a training set of stacked CSI vectors. From $\mathcal D_{\rm tr}$, the BS (or a network controller) fits a $K$-component zero-mean proper complex Gaussian-mixture
\begin{equation}
\label{eq:gmm_offline_learning}
p(\mathbf h)
=
\sum_{c=1}^{K}\pi_c\,\CN(\mathbf h;\mathbf 0,\mathbf R_c),
\end{equation}
where $\pi_c=\Pr(C=c)$ is the prior probability of geometry state $c$ and $\mathbf R_c$ is its covariance. Estimating $\{\pi_c,\mathbf R_c\}_{c=1}^{K}$ is a standard maximum-likelihood task; a common choice is the expectation--maximization (EM) algorithm \cite{EM}. Given current parameters, in the $i$th iteration, the E-step computes the posterior responsibilities
\begin{equation}
r_{i,c}
\triangleq
\Pr(C=c\mid \mathbf h^{(i)})
=
\frac{\pi_c\,\CN(\mathbf h^{(i)};\mathbf 0,\mathbf R_c)}
{\sum_{j=1}^{K}\pi_j\,\CN(\mathbf h^{(i)};\mathbf 0,\mathbf R_j)},
\end{equation}
and the M-step updates
\begin{equation}
\pi_c \leftarrow \frac{1}{N_{\rm tr}}\sum_{i=1}^{N_{\rm tr}} r_{i,c},
\qquad
\mathbf R_c \leftarrow
\frac{\sum_{i=1}^{N_{\rm tr}} r_{i,c}\,\mathbf h^{(i)}(\mathbf h^{(i)})^{\mathsf H}}
{\sum_{i=1}^{N_{\rm tr}} r_{i,c}}.
\end{equation}
Each iteration increases the data likelihood and converges to a stationary point; in practice, $k$-means initialization is often used to avoid poor local optima.

Once the covariances are learned, the BS computes for each state the KLT basis via eigendecomposition
\begin{equation}
\label{eq:gmtc_component_eig}
\mathbf R_c
=
\mathbf U_c \boldsymbol{\Lambda}_c \mathbf U_c^{\mathsf H},
\qquad
\boldsymbol{\Lambda}_c
=
\diag(\lambda_{c,1},\ldots,\lambda_{c,N}),
\end{equation}
with $\lambda_{c,1}\ge \cdots \ge \lambda_{c,N}\ge 0$. The resulting shared transform dictionary is
\begin{equation}
\label{eq:gmtc_dictionary}
\mathcal B
\triangleq
\{(\pi_c,\mathbf U_c,\boldsymbol{\Lambda}_c)\}_{c=1}^{K}.
\end{equation}
Because $\mathcal B$ captures long-term second-order statistics, it can be distributed and stored on a slow timescale (e.g., per deployment region or via periodic updates). Online CSI feedback therefore needs only a short state label and the quantized transform coefficients; no per-block covariance or basis signaling is required.


\subsection{Online GMTC Encoding and Decoding}
\label{subsec:online_gmtc}

Consider one CSI realization $\mathbf h\in\C^N$ to be compressed. Let $R_{\rm tot}$ denote the total feedback budget in \emph{bits per complex channel dimension}. GMTC produces a feedback bitstream
\begin{equation}
s=(s_{\rm lbl},s_{\rm q}),
\end{equation}
where $s_{\rm lbl}$ encodes the state label (losslessly) and $s_{\rm q}$ encodes the quantized transform coefficients (lossy quantization + lossless entropy coding).

\paragraph*{1) State Inference at UE}
Using the learned mixture \eqref{eq:gmm_offline_learning}, the UE computes the posterior score $p(C=c\mid \mathbf h)$ 
and selects the MAP label
\begin{equation}
\label{eq:map_state_selection}
\hat c
=
\arg\max_{c\in\{1,\ldots,K\}}
p(C=c\mid \mathbf h).
\end{equation}

\paragraph*{2) Lossless Label Signaling}
The label $\hat c$ is entropy-coded using a lossless source code matched to $\{\pi_c\}$. If the geometry state remains constant for $\tau$ consecutive CSI feedback blocks, then the normalized label rate is
\begin{equation}
\label{eq:label_rate_gmtc}
R_{\rm lbl}
=
\frac{H(C)}{\tau N}
\quad \text{bits per complex channel dimension},
\end{equation}
where $H(C)=-\sum_{c=1}^{K}\pi_c\log_2\pi_c$.

\paragraph*{3) Component-Matched Transform at UE}
Conditioned on $\hat c$, the UE applies the component-matched KLT
\begin{equation}
\label{eq:gmtc_klt}
\tilde{\mathbf h}
=
\mathbf U_{\hat c}^{\mathsf H}\mathbf h.
\end{equation}
Under the model, the transform coefficients are approximately independent with variances
\begin{equation}
\tilde h_m \sim \CN(0,\lambda_{\hat c,m}),\qquad m=1,\ldots,N.
\end{equation}

\paragraph*{4) Lossy Quantization with RD-Motivated Allocation}
The coefficient-rate budget is
\begin{equation}
\label{eq:gmtc_rate_split}
R_{\rm q}=R_{\rm tot}-R_{\rm lbl}.
\end{equation}
As shown in \cref{sec:fundamental}, the rate/distortion allocation that is optimal for the conditional Gaussian benchmark is governed by a \emph{single global} reverse-waterfilling level $\mu$ (see Theorem~\ref{thm:global_waterlevel_gmm}). Accordingly, the ideal per-mode allocations are
\begin{align}
\label{eq:gmtc_per_mode_distortion}
d_{\hat c,m}^{\star}
&=
\min\{\lambda_{\hat c,m},\mu\},
\\
\label{eq:gmtc_per_mode_rate}
r_{\hat c,m}^{\star}
&=
\left[\log_2\!\left(\frac{\lambda_{\hat c,m}}{\mu}\right)\right]_+.
\end{align}
In practice, GMTC implements \eqref{eq:gmtc_per_mode_rate} by mapping $r_{\hat c,m}^{\star}$ to a finite set of quantizer resolutions and applying scalar quantization to each $\tilde h_m$, producing discrete indices $\{i_m\}_{m=1}^{N}$.

\paragraph*{5) Lossless Entropy Coding of Quantized Indices}
The index sequence $\{i_m\}$ is then entropy-coded to form $s_{\rm q}$. This stage reduces the \emph{realized} feedback rate without introducing additional distortion.

\paragraph*{6) CSI Feedback and BS Reconstruction}
We assume the uplink supports reliable delivery of the feedback bitstream (e.g., via suitable channel coding and adaptive modulation). Upon receiving $s=(s_{\rm lbl},s_{\rm q})$, the BS decodes $\hat c$, retrieves $\mathbf U_{\hat c}$ from $\mathcal B$, reconstructs $\hat{\tilde{\mathbf h}}$, and applies the inverse KLT
\begin{equation}
\label{eq:gmtc_inverse_klt}
\hat{\mathbf h}
=
\mathbf U_{\hat c}\hat{\tilde{\mathbf h}}.
\end{equation}


\subsection{Interpretation and Complexity}
\label{subsec:gmtc_discussion}

GMTC is a hybrid of state selection and TC where the only nonlinear step is the selection of a covariance state; conditioned on the selected state, encoding and decoding reduce to classical KLT TC with coefficient-wise quantization. Thus the UE performs one $N$-dimensional transform and $N$ scalar quantizations per CSI block, i.e., the same order of online complexity as single covariance TC (and $O(N\log N)$ when a fast transform is available). The additional cost is \emph{offline} storage of the dictionary, which scales as $O(KN^2)$ for dense bases and can be reduced using structured covariance models. Because each conditional component is Gaussian, the RD behavior of GMTC is analytically predictable: \cref{sec:fundamental} shows that the optimal allocation is governed by a single multi-modal reverse-waterfilling level (Theorem~\ref{thm:global_waterlevel_gmm}) and that label-aware GMTC achieves the conditional benchmark up to a normalized label-entropy term (Theorem~\ref{thm:oracle_achievable}).

\section{Rate-Distortion Bounds}
\label{sec:fundamental}

Having described GMTC in \cref{sec:gmtc}, we now provide an information-theoretic RD characterization for Gaussian-mixture CSI. The exact RD function of a mixture source is generally intractable because the decoder does not observe the latent geometry state. Guided by the entropy decomposition in Lemma~\ref{lemma2}, we derive a sharp sandwich bound: a genie-aided converse lower bound obtained when the state is revealed for free, and a label-aware achievable upper bound obtained by explicitly transmitting the state label. Both bounds are governed by a \emph{single reverse-waterfilling level} shared across all eigenmodes of all mixture components.

 \subsection{RD Function for Single-State Gaussian}
\label{subsec:single_gaussian_rd}

We begin with the classical benchmark in which the stacked CSI vector $\mathbf h\in\C^N$ follows a single proper complex Gaussian law, i.e., 
$\mathbf h \sim \CN(\mathbf 0,\mathbf R).$
Let the eigendecomposition be
\begin{equation}
\mathbf R=\mathbf U\boldsymbol{\Lambda}\mathbf U^{\mathsf H},
\qquad
\boldsymbol{\Lambda}=\diag(\lambda_1,\ldots,\lambda_N),
\end{equation}
with $\lambda_1\ge \cdots \ge \lambda_N\ge 0$.

\begin{proposition}[RD function of correlated Gaussian CSI]
\label{prop:gaussian_csi_rd}
Under the distortion measure $D=\frac{1}{N}\E[\|\mathbf h-\hat{\mathbf h}\|_2^2]$, the RD function is
\begin{align}
\label{eq:single_gaussian_rate}
R_{\rm G}(D)
&=
\frac{1}{N}\sum_{m=1}^{N}
\left[\log_2\!\left(\frac{\lambda_m}{\mu}\right)\right]_+,
\\
\label{eq:single_gaussian_dist}
D
&=
\frac{1}{N}\sum_{m=1}^{N}\min\{\lambda_m,\mu\},
\end{align}
for some water level $\mu\ge 0$. Equivalently, the optimal per-mode distortion allocation is
\begin{equation}
d_m^\star=\min\{\lambda_m,\mu\}, \qquad m=1,\ldots,N.
\end{equation}
\end{proposition}
\begin{proof}
    See \cite{goyal2002theoretical,cover2006elements}.
\end{proof}

\subsection{Multi-Modal Reverse-Waterfilling for Gaussian-Mixture}
\label{subsec:conditional_benchmark}

We now return to the Gaussian-mixture CSI source
\begin{equation}
\label{eq:gmm_csi_rd}
p({\mathbf h})
=
\sum_{c=1}^{K}\pi_c\,\CN({\mathbf h};\mathbf 0,\mathbf R_c),
\end{equation}
where $\Pr(C=c)=\pi_c$, $\pi_c>0$, and $\sum_{c=1}^{K}\pi_c=1$. Let
\begin{equation}
\mathbf R_c
=
\mathbf U_c \boldsymbol{\Lambda}_c \mathbf U_c^{\mathsf H},
\end{equation}
with
$\boldsymbol{\Lambda}_c
=
\diag(\lambda_{c,1},\ldots,\lambda_{c,N}).$  Conditioned on $C=c$, $\mathbf h$ is a correlated Gaussian vector with covariance $\mathbf R_c$. From the classical RD theory \cite{cover2006elements}, its RD function admits the parametric form
\begin{align}
\label{eq:Rc_rate}
R_c(D_c)
&=
\frac{1}{N}\sum_{m=1}^{N}
\left[\log_2\!\left(\frac{\lambda_{c,m}}{\mu_c}\right)\right]_+,
\\
\label{eq:Rc_dist}
D_c
&=
\frac{1}{N}\sum_{m=1}^{N}\min\{\lambda_{c,m},\mu_c\},
\end{align}
for some water level $\mu_c\ge 0$.

If both terminals know $C$, they may allocate a component-dependent distortion budget $\{D_c\}$ and code each component optimally. This fact induces the RD function:
\begin{equation}
\label{eq:Rcond_def}
R_{\rm cond}(D)
\triangleq
\min_{\{D_c\ge 0\}:\ \sum_{c=1}^{K}\pi_c D_c\le D}
\sum_{c=1}^{K}\pi_c R_c(D_c).
\end{equation}
Since each $R_c(D_c)$ is convex and nonincreasing in $D_c$, \eqref{eq:Rcond_def} is a convex program.

\begin{theorem}[Single global water level]
\label{thm:global_waterlevel_gmm}
The optimizer of \eqref{eq:Rcond_def} is characterized by a single global water level $\mu$ such that
\begin{align}
\label{eq:gmm_global_rate}
R_{\rm cond}(D)
&=
\frac{1}{N}\sum_{c=1}^{K}\pi_c\sum_{m=1}^{N}
\left[\log_2\!\left(\frac{\lambda_{c,m}}{\mu}\right)\right]_+,
\\
\label{eq:gmm_global_dist}
D
&=
\frac{1}{N}\sum_{c=1}^{K}\pi_c\sum_{m=1}^{N}
\min\{\lambda_{c,m},\mu\}.
\end{align}
\end{theorem}

\begin{proof}
We derive the first-order optimality condition, i.e., Karush–Kuhn–Tucker (KKT) condition.

Fix a component $c$. For any $D_c$ with $0\le D_c<\frac{1}{N}\sum_{m=1}^{N}\lambda_{c,m}$, the corresponding water level $\mu_c$ in \eqref{eq:Rc_dist} is uniquely determined because $D_c(\mu_c)=\frac{1}{N}\sum_m \min\{\lambda_{c,m},\mu_c\}$ is continuous, nondecreasing, and strictly increasing on intervals where the active set $\{m:\lambda_{c,m}>\mu_c\}$ is fixed.

Define the active set $\mathcal A_c(\mu_c)\triangleq\{m:\lambda_{c,m}>\mu_c\}$ and its cardinality $K_c(\mu_c)=|\mathcal A_c(\mu_c)|$. On any interval where $\mathcal A_c(\mu_c)$ is constant,
\begin{equation}
\begin{aligned}
R_c(\mu_c)
&=
\frac{1}{N}\sum_{m\in\mathcal A_c(\mu_c)}
\log_2\!\left(\frac{\lambda_{c,m}}{\mu_c}\right),\\
D_c(\mu_c)
&=
\frac{1}{N}\Big(\sum_{m\notin\mathcal A_c(\mu_c)}\lambda_{c,m}+K_c(\mu_c)\mu_c\Big).
\end{aligned}
\label{eq:Rc_Dc_mu}
\end{equation}
Differentiating yields
\begin{equation}
\frac{dR_c}{d\mu_c}
=
-\frac{K_c(\mu_c)}{N\mu_c\ln 2},
\qquad
\frac{dD_c}{d\mu_c}
=
\frac{K_c(\mu_c)}{N}.
\end{equation}
Hence, wherever differentiable,
\begin{equation}
\label{eq:dRc_dDc}
\frac{dR_c}{dD_c}
=
\frac{dR_c/d\mu_c}{dD_c/d\mu_c}
=
-\frac{1}{\mu_c\ln 2}.
\end{equation}
At points where the active set changes (i.e., $\mu_c$ crosses an eigenvalue), the derivative may not exist, but \eqref{eq:dRc_dDc} still characterizes the subgradient; this is sufficient for KKT optimality in a convex program.

Consider the Lagrangian
\begin{align}
  \mathcal L(\{D_c\},\nu,\{\eta_c\})
&=
\sum_{c=1}^{K}\pi_c R_c(D_c)
+
\nu\Big(\sum_{c=1}^{K}\pi_c D_c-D\Big)\nonumber\\
&-\sum_{c=1}^{K}\eta_c D_c,  
\end{align}
with multipliers $\nu\ge 0$ and $\eta_c\ge 0$. At an optimum, we have the KKT conditions:
(i) primal feasibility: $D_c\ge 0$ and $\sum_c \pi_c D_c\le D$;
(ii) dual feasibility: $\nu\ge 0$, $\eta_c\ge 0$;
(iii) complementarity: $\nu(\sum_c \pi_c D_c-D)=0$ and $\eta_c D_c=0$;
(iv) stationarity: for each $c$,
\begin{equation}
\label{eq:KKT_stationarity}
\pi_c \frac{dR_c}{dD_c} + \nu\pi_c - \eta_c = 0,
\end{equation}
where $\frac{dR_c}{dD_c}$ is interpreted as a subgradient when nondifferentiable.

For any \emph{active} component with $D_c>0$, complementarity forces $\eta_c=0$, so \eqref{eq:KKT_stationarity} reduces to
\begin{equation}
\frac{dR_c}{dD_c} + \nu = 0.
\end{equation}
Using \eqref{eq:dRc_dDc}, this implies
\begin{equation}
-\frac{1}{\mu_c\ln 2} + \nu = 0
\quad\Rightarrow\quad
\mu_c = \frac{1}{\nu\ln 2},
\end{equation}
which is independent of $c$. Thus all active components share a common water level $\mu$. Substituting $\mu_c=\mu$ into the componentwise formulas \eqref{eq:Rc_rate}--\eqref{eq:Rc_dist} and averaging over $c$ yields \eqref{eq:gmm_global_rate}--\eqref{eq:gmm_global_dist}, completing the proof.
\end{proof}

Theorem~\ref{thm:global_waterlevel_gmm} implies that, although the source is a \emph{mixture} of $K$ different Gaussian geometries, the optimal RD allocation does \emph{not} require $K$ separate water levels. There is only one knob: a single water level $\mu$ shared by \emph{every} eigenmode of \emph{every} component. Operationally, we can imagine collecting all eigenmodes $\{\lambda_{c,m}\}$ into one big pool, where each mode from component $c$ is weighted by how often that component occurs (its probability $\pi_c$). Reverse-waterfilling is then performed once on this pooled spectrum: keep the modes with $\lambda_{c,m}>\mu$ and \textit{fill} them to distortion $\mu$, while modes with $\lambda_{c,m}\le \mu$ are left untouched and contribute their full variance as distortion.

This has two important consequences. First, the mixture structure does not destroy the simplicity of Gaussian TC: the entire conditional benchmark $R_{\rm cond}(D)$ is still parameterized by a \emph{single scalar} $\mu$. Second, the mixture probabilities $\{\pi_c\}$ enter only through averaging in \eqref{eq:gmm_global_rate}--\eqref{eq:gmm_global_dist}, which means the bit budget is automatically steered toward the states that occur most often, without requiring the encoder to solve a separate allocation problem for each state. In particular, once $\mu$ is chosen to meet a target rate (or distortion), GMTC can implement the corresponding per-mode allocation in the selected state $\hat c$ by the same reverse-waterfilling rule as the single-Gaussian case, but with a water level that is globally consistent across all states.

\subsection{Genie-Aided Converse Lower Bound}
\label{subsec:genie_lb}
To obtain a converse benchmark, we consider a genie that reveals the geometry state $C$ to both the UE encoder and the BS decoder. With $C$ available as common side information, the CSI source becomes conditionally Gaussian, and the RD problem reduces to the conditional benchmark $R_{\rm cond}(D)$ in \cref{subsec:conditional_benchmark} (with the multi-modal reverse-waterfilling structure in Theorem~\ref{thm:global_waterlevel_gmm}). Since side information cannot increase the minimum required rate, this yields a lower bound on the true mixture RD function.

\begin{theorem}[Genie-aided lower bound]
\label{thm:genie_lb}
Let $R^\star(D)$ denote the exact RD function of the Gaussian-mixture CSI source in \eqref{eq:gmm_csi_rd}. If the geometry state $C$ is revealed for free to both the encoder and decoder, then the minimum required rate is exactly $R_{\rm cond}(D)$. Consequently,
\begin{equation}
\label{eq:genie_lower_bound}
R^\star(D)\ge R_{\rm cond}(D).
\end{equation}
\end{theorem}

\begin{proof}
Let $R_{\mathbf h|C}(D)$ denote the RD function when $C$ is available to both terminals at no rate cost:
\begin{equation}
R_{\mathbf h|C}(D)
=
\frac{1}{N}
\inf_{p(\hat{\mathbf h}|\mathbf h,C):
\ \frac{1}{N}\E\|\mathbf h-\hat{\mathbf h}\|_2^2\le D}
I(\mathbf h;\hat{\mathbf h}\mid C).
\end{equation}
Since providing extra side information cannot increase the minimum required rate, we have the standard inequality
\begin{equation}
R^\star(D)\ge R_{\mathbf h|C}(D).
\end{equation}

Next, expand the conditional mutual information via the law of total expectation:
\begin{equation}
I(\mathbf h;\hat{\mathbf h}\mid C)
=
\sum_{c=1}^{K}\pi_c\, I(\mathbf h;\hat{\mathbf h}\mid C=c),
\end{equation}
and similarly decompose the distortion:
\begin{equation}
\frac{1}{N}\E\|\mathbf h-\hat{\mathbf h}\|_2^2
=
\sum_{c=1}^{K}\pi_c\,
\frac{1}{N}\E\!\left[\|\mathbf h-\hat{\mathbf h}\|_2^2\mid C=c\right].
\end{equation}
Define $D_c\triangleq \frac{1}{N}\E[\|\mathbf h-\hat{\mathbf h}\|_2^2\mid C=c]$ so that $\sum_c \pi_c D_c\le D$. Conditioning on $C=c$, the source is Gaussian with covariance $\mathbf R_c$. As a result, by Theorem \ref{thm:global_waterlevel_gmm}, we obtain
\begin{align}
 R_{\mathbf h|C}(D)
&=
\min_{\{D_c\ge 0\}:\ \sum_c \pi_c D_c\le D}
\sum_{c=1}^{K}\pi_c R_c(D_c)\nonumber\\
&=
R_{\rm cond}(D),   
\end{align}
which implies \eqref{eq:genie_lower_bound}.
\end{proof}

The genie-aided lower bound in \Cref{thm:genie_lb} is the cleanest benchmark one can hope for: if the geometry state $C$ were known to \emph{both} the UE and the BS, then the mixture source becomes conditionally Gaussian and the fundamental rate collapses to the conditional RD function,
\begin{equation}
R_{\mathbf h|C}(D)=R_{\rm cond}(D).
\end{equation}
So $R_{\rm cond}(D)$ is the best possible curve when the slowly varying propagation regime is common knowledge.

A more realistic intermediate regime is decoder-only side information: the BS may infer $C$ (perhaps imperfectly) from long-term context such as uplink tracking of large-scale geometry parameters. This is a Wyner--Ziv setting \cite{WynerZiv}, and the corresponding RD function satisfies the standard ordering
\begin{equation}
R_{\rm cond}(D)\;=\;R_{\mathbf h|C}(D)\;\le\;R_{\rm WZ}(D)\;\le\;R^\star(D),
\end{equation}
i.e., not knowing $C$ at the encoder can only make compression harder.

Finally, label-aware GMTC can be viewed as a simple way to \emph{manufacture} common side information: it explicitly sends the state index, at a normalized cost $H(C)/(\tau N)$. If the BS already has correlated information about $C$, then the label overhead can, in principle, be reduced to the residual uncertainty of the state at the decoder (a conditional-entropy type cost rather than $H(C)$). In the limiting case where the BS essentially knows $C$, the label cost vanishes and the operational target is exactly the genie-aided curve $R_{\rm cond}(D)$.

\subsection{Label-Aware GMTC Achievability Upper Bound}
\label{subsec:oracle_ub}

We now derive an achievable RD bound using the proposed GMTC scheme. In this bound, we assume that the geometry-state label is compressed losslessly and is perfectly shared between the encoder at the UE and the decoder at the BS.

\begin{theorem}[Label-aware GMTC achievable upper bound]
\label{thm:oracle_achievable}
Assume that the latent state remains constant over $\tau$ consecutive CSI feedback intervals and that its label is fed back losslessly once per state-coherence interval. Then the following rate is achievable:
\begin{equation}
\label{eq:gmtc_achievable_bound}
R_{\rm GMTC}(D)
=
R_{\rm cond}(D)
+
\frac{H(C)}{\tau N},
\end{equation}
where $H(C)=-\sum_{c=1}^{K}\pi_c \log_2 \pi_c$.
\end{theorem}

\begin{proof}
We construct an explicit coding scheme and track its rate. Consider a blocklength-$n$ code where the latent state is constant over groups of $\tau$ CSI vectors. For simplicity, we assume $n$ is a multiple of $\tau$ and write $n=\tau L$. Let $C_\ell$ denote the state in group $\ell\in\{1,\dots,L\}$. The encoder transmits the sequence $(C_1,\dots,C_L)$ using an optimal lossless code. By the source coding theorem, for any $\epsilon>0$ and sufficiently large $L$, the label sequence can be transmitted with expected codelength at most $L(H(C)+\epsilon)$ bits. Normalizing by $nN=\tau L N$ gives a normalized label rate of at most $\frac{H(C)+\epsilon}{\tau N}$ bits per complex dimension.

Fix any distortion allocation $\{D_c\}$ satisfying $\sum_c \pi_c D_c\le D$. For each component $c$, by Shannon's RD theorem applied to the Gaussian source $\CN(\mathbf 0,\mathbf R_c)$, there exist (for sufficiently large $n_c$) codes that achieve rate arbitrarily close to $R_c(D_c)$ while meeting distortion $D_c$. Since the decoder knows $C_\ell$ after decoding the label, encoder and decoder may use the corresponding component-matched Gaussian transform code within each group. Averaging over groups, the expected per-letter coefficient rate is at most $\sum_{c=1}^{K}\pi_c (R_c(D_c)+\epsilon)$.

Combining the two parts, the total achievable rate satisfies
\begin{equation}
R \le \frac{H(C)+\epsilon}{\tau N} + \sum_{c=1}^{K}\pi_c (R_c(D_c)+\epsilon).
\end{equation}
Letting $n\to\infty$ (so $\epsilon\to 0$) and minimizing over all feasible $\{D_c\}$ yields
\begin{equation}
R_{\rm GMTC}(D)
\le
\frac{H(C)}{\tau N} + R_{\rm cond}(D),
\end{equation}
which proves \eqref{eq:gmtc_achievable_bound}.
\end{proof}

\subsection{Gap Analysis}
\label{subsec:consequences}

The following corollary illustrates the conditions under which the bounds coincide.

 \begin{corollary}[Vanishing entropy gap]
\label{cor:gap_vanish}
For every distortion level $D$, the exact RD function satisfies
\begin{equation}
\label{eq:fundamental_sandwich}
R_{\rm cond}(D)
\le
R^\star(D)
\le
R_{\rm GMTC}(D).
\end{equation}
Hence, if the mixture order $K$ and the geometry coherence length $\tau$ are fixed, then as $N\rightarrow \infty$,
\begin{equation}
\label{eq:gap_vanish_simple}
0
\le
R_{\rm GMTC}(D)-R_{\rm cond}(D)
=
\frac{H(C)}{\tau N}
\le
\frac{\log_2 K}{\tau N}
\to 0.
\end{equation}
Therefore, in the large-dimensional regime, the upper and lower bounds coincide asymptotically, and GMTC is asymptotically optimal.
\end{corollary}


\begin{proof}
The proof is direct from Theorem \ref{thm:genie_lb} and Theorem \ref{thm:oracle_achievable}.
\end{proof}

Corollary~\ref{cor:gap_vanish} decomposes the feedback cost into a \emph{lossy} part and a \emph{lossless} part. The lossy part is $R_{\rm cond}(D)$, the conditional Gaussian RD function obtained when the active covariance state is known and component-matched TC is used. The lossless part is the normalized label rate $H(C)/(\tau N)$ needed to specify that state. Hence multi-modality does not force per-block covariance or basis feedback; it only introduces the label overhead. Since $N=N_tN_c$ is large in wideband massive MIMO--OFDM, the label term is typically negligible, and for fixed $K$ and $\tau$ the bounds become tight as $N\to\infty$. Increasing the mixture size improves modeling accuracy but increases the worst-case overhead only logarithmically through $H(C)\le \log_2 K$.

\section{Simulation Result}
\label{sec:simulation}

In this section, we validate the performance of the proposed GMTC framework via two sets of numerical experiments. First, we utilize synthetic CSI samples generated from a Gaussian-mixture distribution to assess the fundamental gains of the proposed framework. Second, we evaluate the performance of GMTC framework using the COST2100 dataset \cite{COST2100} to validate its effectiveness in realistic CSI environments. Furthermore, to assess its viability for deployment in FDD systems, we provide a comprehensive analysis of resource consumption, including model size, UE-side inference complexity, and online memory overhead.


\begin{table*}[t]
\centering
\caption{Parameter count, online memory, and encoder FLOPs for the synthetic Gaussian-mixture experiment for $N=64$.}
\label{tab:synthetic_complexity_full}
\begin{tabular}{llrccc}
\toprule
Method & $K$ & Scale & \# Parameters & Online Memory (Byte) & Encoder FLOPs \\
\midrule
TC              & -- & -- & $2N^2 + 2N + C^{\ast}$         & $8N^2$   & $N^2 + 4N$ \\
GMTC & -- & -- & $2KN^2 + 2KN + C^{\ast}$        & $8KN^2$  & $N^2 +  4N$ \\
\midrule
\multirow{9}{*}{Swin-ViT-NTC$^{\dagger\ddagger}$}
  & \multirow{3}{*}{$8$}
    & $1\times $  & $32\times10^3$   & $9.9\times10^6$  & $215\times10^3$ \\
  & & $8\times $  & $266\times10^3$  & $10.8\times10^6$ & $1.8\times10^6$ \\
  & & $64\times $ & $1.8\times10^6$  & $17.1\times10^6$ & $12.8\times10^6$ \\
\cmidrule{2-6}
  & \multirow{3}{*}{$16$}
    & $1\times $  & $55\times10^3$   & $9.9\times10^6$  & $375\times10^3$ \\
  & & $8\times $  & $634\times10^3$  & $12.3\times10^6$ & $4.3\times10^6$ \\
  & & $64\times $ & $4.1\times10^6$  & $26.3\times10^6$ & $28.7\times10^6$ \\
\cmidrule{2-6}
  & \multirow{3}{*}{$64$}
    & $1\times $  & $266\times10^3$  & $10.8\times10^6$ & $1.8\times10^6$ \\
  & & $8\times $  & $1.8\times10^6$  & $17.1\times10^6$ & $12.8\times10^6$ \\
  & & $64\times $ & $12.9\times10^6$ & $61.7\times10^6$ & $90.5\times10^6$ \\
\midrule
\multicolumn{3}{l}{\textbf{Ratio (NTC 64$\times$ / GMTC, $K=64$)}}
  & $\mathbf{24\times}$
  & $\mathbf{29\times}$
  & $\mathbf{20795\times}$ \\
\bottomrule
\multicolumn{6}{l}{$^{\ast}$ $C=2{,}000$ denotes the size of the uniform quantization lookup table shared across all rate points.} \\
\multicolumn{6}{l}{$^\dagger$ Parameter count is per rate-point model.} \\
\multicolumn{6}{l}{$^\ddagger$ Online memory is measured via \texttt{torch.cuda.max\_memory\_allocated()} and converted to bytes;} \\
\multicolumn{6}{l}{\phantom{$^\ddagger$} encoder FLOPs are measured via \texttt{torch.profiler}.}
\end{tabular}
\end{table*}

\begin{figure*}[t]
\centering
\includegraphics[width=\textwidth]{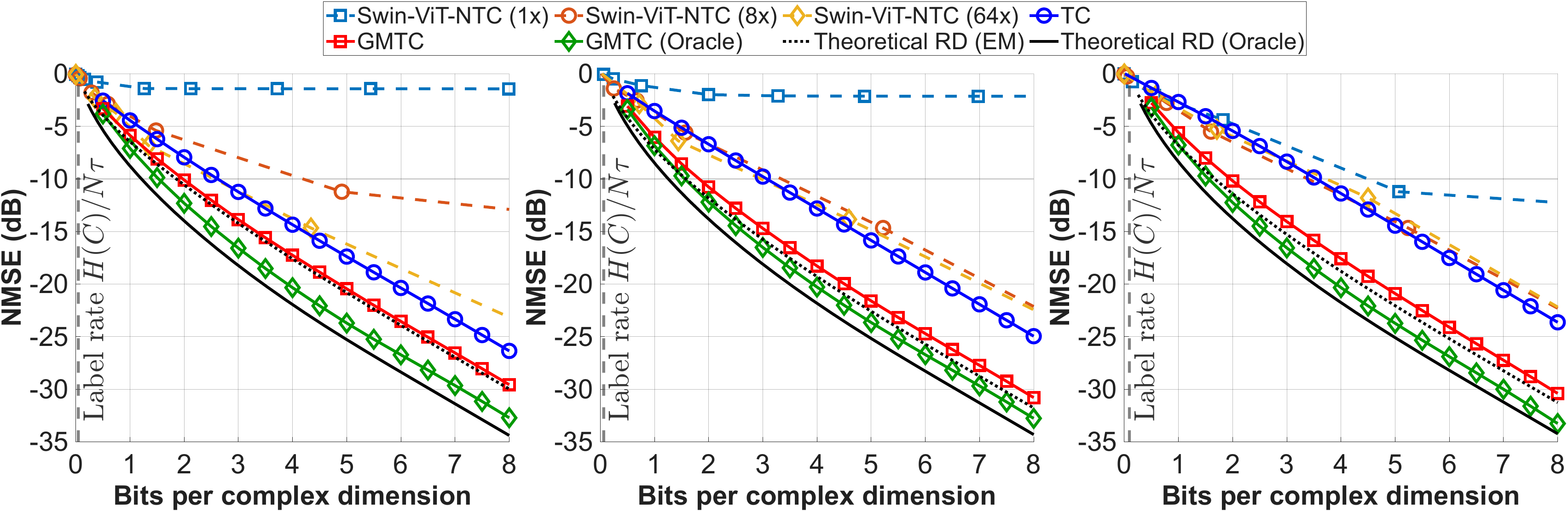}

\vspace{0.3em}
\makebox[\textwidth][l]{%
  \hspace*{1.2em}%
  \makebox[0.32\textwidth][c]{\hspace*{0.5em}(a) $N=64$, $K=8$}%
  \makebox[0.34\textwidth][c]{\hspace*{1.0em}(b) $N=64$, $K=16$}%
  \makebox[0.32\textwidth][c]{\hspace*{1.3em}(c) $N=64$, $K=64$}%
}

\caption{The RD performance validation on synthetic Gaussian-mixture data. Panels (a), (b), and (c) correspond to $N=64$ with $K=8$, $K=16$, and $K=64$, respectively. The oracle benchmarks are invariant to $K$ as the conditional RD limit $R_{\mathrm{cond}}(D)$ is preserved under the uniform state distribution.} 
\label{fig:gmm_rd_validation}
\end{figure*}

\subsection{Performance Metrics and Benchmark Schemes}
To evaluate the RD performance, we employ NMSE and bits per complex dimension as the distortion and rate metrics, respectively, as defined in Sec. \ref{sec:rd_formulation}. 
The performance of the proposed framework is evaluated against the following benchmark schemes:

\begin{itemize}[leftmargin=*,itemsep=4pt]
\item \textbf{Theoretical RD (Oracle)}: The analytical achievable upper bound $R_{\mathrm{GMTC}}(D)$ in \eqref{eq:gmtc_achievable_bound}, computed using the \textit{ground-truth mixture parameters}. This serves as the fundamental architectural limit of the label-aware GMTC.

\item \textbf{Theoretical RD (EM)}: The analytical bound $R_{\mathrm{GMTC}}(D)$ computed using the \textit{EM-estimated parameters}. The gap between this and the oracle curve quantifies the performance loss due to offline modeling errors.

\item \textbf{TC}: Classical single-Gaussian TC based on a single global covariance model, representing the performance under a unimodal assumption.

\item \textbf{GMTC}: The practical, finite-blocklength implementation of the mixture-aware TC using \textit{offline training}, \textit{MAP state selection}, and \textit{entropy-coded KLT coefficients}.

\item \textbf{Oracle-label GMTC}: An empirical benchmark for the GMTC where the \textit{ground-truth mixture parameters} and \textit{state label} are perfectly provided to both the encoder and decoder. 

\item \textbf{Swin-ViT-NTC}: A state-of-the-art (SOTA) transformer-based NTC baseline introduced in \cite{park2025transformer}.

\item \textbf{CsiNet}: A representative convolutional neural network based autoencoder baseline for CSI compression \cite{deeplearning2018}, exclusively for the COST2100 dataset evaluation. 
\end{itemize}

To evaluate the performance across different neural network model capacities, we \textit{scale} the Swin-ViT-NTC baseline by varying its feature dimension. Rather than an arbitrary selection, this scaling strategy was determined through extensive hyperparameter sweeps involving latent dimensions, layer depths, and kernel sizes. Among these, adjusting the feature dimension yielded the most significant performance gains relative to the increase in parameter count. This approach provides a robust basis for benchmarking the NTC model against GMTC under comparable parameter budgets. 

As noted in \cite{park2025transformer}, the Swin-ViT-NTC baseline requires training a separate model for each RD point to achieve its peak performance. While variable rate CSI compression methods\cite{park2024multi,park2025transformer,Liotopoulos2025OFSQ} offer greater flexibility, they typically suffer from a performance degradation due to the inherent trade-off between dedicated optimization and model generalization \cite{liang2022changeable}. Since our primary objective is to evaluate the proposed GMTC framework against SOTA RD limit, we adopt the dedicated Swin-ViT-NTC as our benchmark. Due to space constraints, the detailed architecture, loss functions, and training hyperparameters for Swin-ViT-NTC are described in \cite{park2025transformer}.

 \begin{figure}[t]
    \centering
    \includegraphics[width=0.45\textwidth]{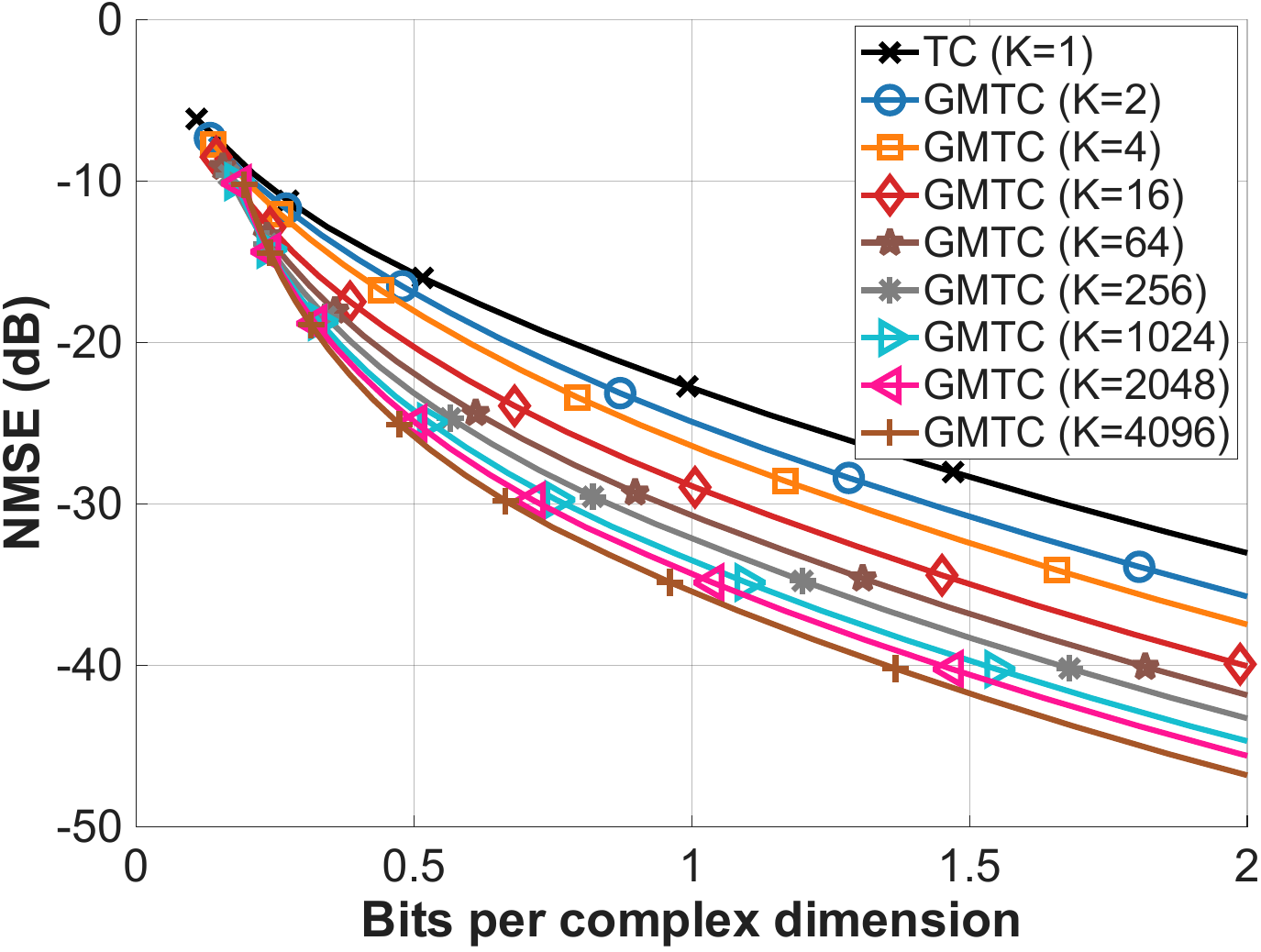}
    \caption{The RD performance of the GMTC scheme when $M=128$ and $K\in \{1,2,4,16,64,256,1024,2048,4096\}$. }
    \label{fig:cost_k_sweep_n128}
    \vspace{-4mm}
\end{figure}

\begin{figure}[t]
    \centering
    \includegraphics[width=0.45\textwidth]{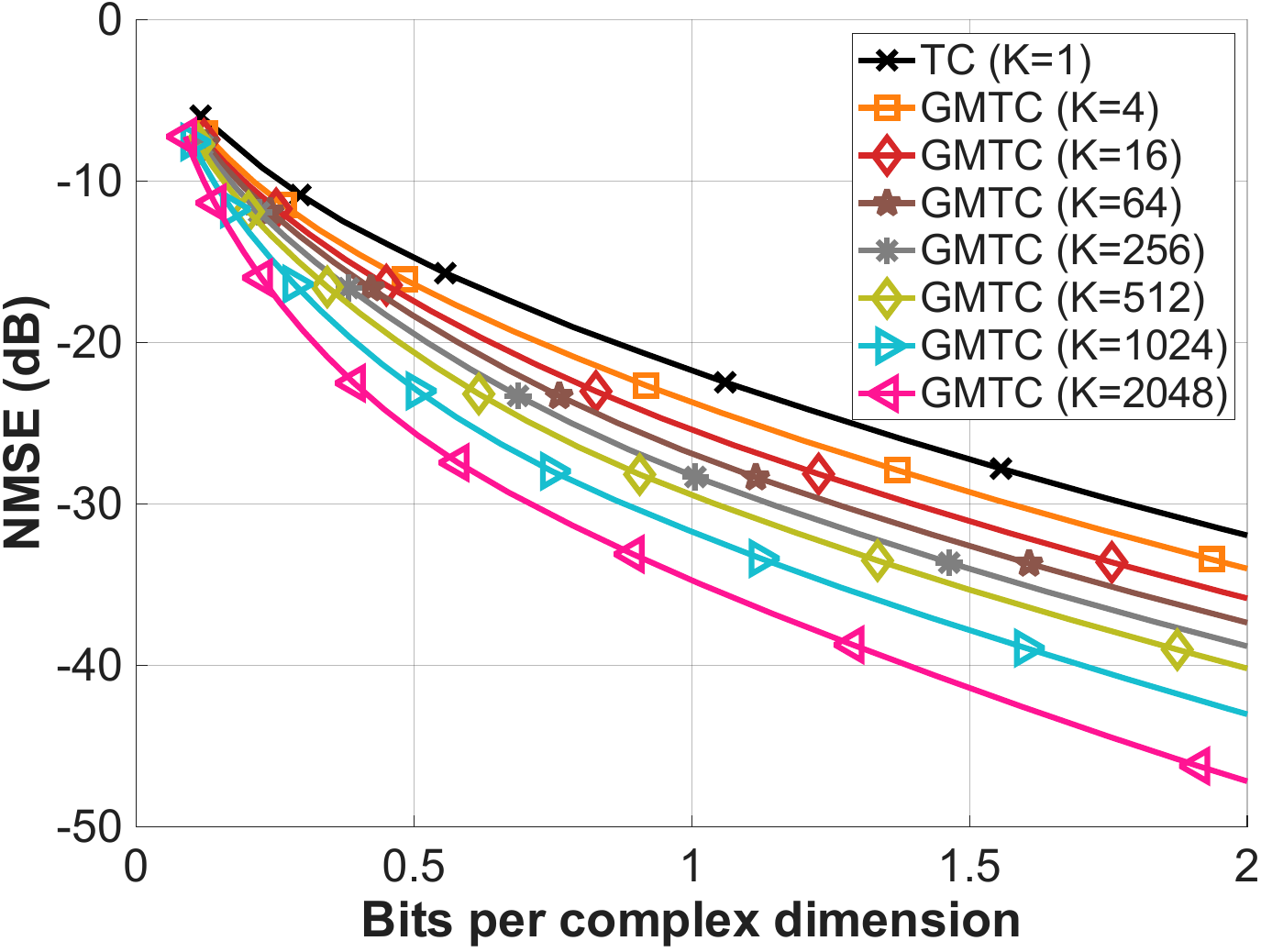}
    \caption{The RD performance of the GMTC scheme when $M=512$ and $K\in \{1,4,16,64,256, 512,1024,2048\}$. }
    \label{fig:cost_k_sweep_n512}
    \vspace{-4mm}
\end{figure}

\subsection{Synthetic Gaussian-Mixture Validation}

We first perform a controlled source-model validation experiment using synthetic Gaussian-mixture data whose structure is deliberately matched to the channel model in Section \ref{sec:csi_channel_model}. Specifically, we generated $K$-state zero-mean proper complex Gaussian-mixture data with $K\in\{8,16,64\}.$ The state probabilities are chosen to be uniform for clarity, i.e., $\pi_c=1/K$. Conditioned on state $C=c$, the synthetic CSI vector is generated as
$\mathbf h \mid (C=c)\sim \CN(\mathbf 0,\mathbf R_c),$
where the covariance is constructed from the same eigenvectors yet with distinct eigenvalues as $\mathbf R_c
= \mathbf F\boldsymbol{\Lambda}_c\mathbf F^{\mathsf H}$, where ${\bf F}\in \mathbb{C}^{N\times N}$ is the Fourier matrix and $\boldsymbol{\Lambda}_c\triangleq \mathrm{diag}(\lambda_{c,1},\dots,\lambda_{c,64})$ whose element $\lambda_{c,m}$ is drawn i.i.d. from a log-uniform distribution, $\lambda_{c,m}\triangleq 10^{\delta_{c,m}},~\mathrm{where}~\delta_{c,m}\sim\mathrm{Unif}[-2,2]$ for $c\in [K]$. 

For each value of $K$, we generate a disjoint training dataset of size $N_{\rm tr}=80{,}000$ and a test dataset of size $N_{\rm te}=20{,}000$. All benchmarks--including the proposed GMTC, the theoretical RD (EM), TC, and the Swin-ViT-NTC baseline-- utilize this training set for offline learning. Specifically, the BS learns the $K$-component Gaussian-mixture dictionary for GMTC and the theoretical RD (EM), while TC learns a single global covariance model. Following the standard approach in neural network-based CSI compression \cite{deeplearning2018,Lightweight_CSI_feedback_2021}, the Swin-ViT-NTC baseline processes the complex wideband channel as an augmented real-valued tensor. This is achieved by stacking the real and imaginary parts of the CSI into a $2\times 8\times 8$ input structure.



Fig.~\ref{fig:gmm_rd_validation} shows the NMSE performance versus the  realized feedback bits per complex dimension, with the corresponding schemes identified in the legend. Additionally, the gray dashed vertical line indicates the lossless label overhead, $H(C)/(\tau N)$, which represents the minimum rate cost for state signaling in GMTC and the theoretical RD benchmarks. Three main trends emerge in Fig.~\ref{fig:gmm_rd_validation}.

{\bf Mismatch Penalty of Single-Gaussian TC:} First, the mismatch penalty of single-Gaussian TC increases with $K$. As the source exhibits greater multi-modality, a single covariance averages over heterogeneous component spectra, resulting in systematically suboptimal bit allocation across transform modes. This appears as a widening gap between TC and the theoretical benchmark as $K$ grows.

{\bf Stability of the GMTC-to-Oracle Gap:}
Second, the performance gap introduced by offline training remains significantly stable. Unlike the TC baseline, the distance between the GMTC and the oracle limit remains within a consistent $3 \mathrm{dB}$ margin regardless of $K$. This confirms that EM-based modeling effectively captures the essential spectral diversity, preventing the performance degradation observed in single-Gaussian TC.

{\bf Comparison with Swin-ViT-NTC:} Third, the proposed GMTC consistently outperforms the Swin-ViT-NTC baseline, with the performance gap widening significantly as the source complexity $K$ increases. Fig. \ref{fig:gmm_rd_validation} shows that the $1\times$ NTC variant fails to match even the conventional TC baseline, while the $64\times$ variant barely approaches TC. As highlighted in Table \ref{tab:synthetic_complexity_full}, even when compared against $64\times$ NTC variant--which consumes \textbf{24$\times$ more parameters}, \textbf{29$\times$ more online memory}, and a staggering \textbf{20795$\times$ higher encoder FLOPs}--GMTC maintains a clear and substantial lead in RD performance. This demonstrates that for multi-modal CSI sources, the explicit state-adaptive second-order structure exploited by GMTC is more parameter-efficient and computationally effective than high-capacity nonlinear neural transforms.

 

\subsection{COST2100 Indoor CSI Benchmark}

We next evaluate the proposed framework using the COST2100 indoor channel model at 5.3 GHz with $(N_t,N_c)=(32,32)$\cite{COST2100}. For this dataset, the total dimensionality of the augmented real-valued CSI is $2N = 2N_tN_c=2048$. Similar to the synthetic Gaussian-mixture data, the CSI data is reshaped into a two-channel tensor of size $2 \times 32 \times 32$. This allows the neural networks to exploit the spatial-frequency correlation through their respective convolutional or transformer layers. However, in the proposed GMTC framework, the computational complexity scales quadratically with respect to the input dimension, and training on the full dimension is computationally prohibitive. Therefore, we adopt a practical segmentation strategy and let $M \ll N$ denote the operational block dimension for GMTC. Specifically, the full CSI vector is denoted as
\begin{equation}
\mathbf{x}_{\mathrm{full}} = \begin{bmatrix} \mathrm{vec}(\Re(\mathbf{H}))^{\mathrm{T}}; \ \mathrm{vec}(\Im(\mathbf{H}))^{\mathrm{T}} \end{bmatrix}^{\mathrm{T}} \in \mathbb{R}^{2N},
\end{equation}
and its $i$-th segment $\mathbf{x}^{(i)} \in \mathbb{R}^{M}$ is expressed as
\begin{equation}
\mathbf{x}^{(i)} = \mathbf{x}_{\mathrm{full}} \big[ (i-1)M + 1 : iM \big],\ i = 1, \dots, 2N/M,
\end{equation}
where $M$ is chosen as a divisor of $2N$. This partitioning  enables GMTC to learn its Gaussian-mixture dictionary on segmented $M$-dimensional sub-vectors with tractable computational complexity. 

\begin{table*}[t]
\centering
\caption{Parameter count, online memory, and encoder FLOPs for the COST2100 dataset experiment $M=128, K=64$.}
\label{tab:cost2100_complexity}
\begin{tabular}{lrccc}
\toprule
Method & Scale & \# Parameters & Online Memory (B) & Encoder FLOPs \\
\midrule
TC & -- & $2M^2 + 2M + C^{\ast}$ & $8M^2$ & $M^2 + 4M$ \\
GMTC & -- & $2KM^2 + 2KM + C^{\ast}$ & $8KM^2$ & $M^2 + 4M$ \\
\midrule
\multirow{3}{*}{CsiNet$^{\ddagger}$}
    & $1\times $ & $1.1\times10^6$ & $13\times10^6$ & $1.2\times10^6$ \\
    & $2\times $ & $2\times10^6$ & $17\times10^6$ & $2.1\times10^6$ \\
    & $4\times $ & $4.2\times10^6$ & $25.6\times10^6$ & $4.3\times10^6$ \\
\cmidrule{2-5}
\multirow{3}{*}{Swin-ViT-NTC$^{\dagger\ddagger}$}
    & $1\times $ & $1.1\times10^6$ & $7.2\times10^6$ & $13.5\times10^6$ \\
    & $2\times $ & $2\times10^6$ & $15\times10^6$ & $17.9\times10^6$ \\
    & $4\times $ & $4.2\times10^6$ & $30.3\times10^6$ & $26.5\times10^6$ \\
\midrule
\multicolumn{2}{l}{\textbf{Ratio (NTC 4$\times$ / GMTC, $M=128$, $K=64$)}}
    & $\mathbf{2\times}$
    & $\mathbf{3.6 \times}$
    & $\mathbf{1568\times}$ \\
\bottomrule
\multicolumn{5}{l}{$^{\ast}$ $C=2{,}000$ denotes the size of the uniform quantization lookup table shared across all rate points.} \\
\multicolumn{5}{l}{$^\dagger$ Parameter count is per rate-point model.} \\
\multicolumn{5}{l}{$^\ddagger$ Online memory is measured via \texttt{torch.cuda.max\_memory\_allocated()} and converted to bytes;} \\
\multicolumn{5}{l}{\phantom{$^\ddagger$} encoder FLOPs are measured via \texttt{torch.profiler}.}
\end{tabular}
\end{table*}

Similar to NTC baseline, the CsiNet baseline is also evaluated across various model scales to investigate the impact of model capacity. Specifically, we adopt the original CsiNet encoder and decoder as the base architecture, use a 64-dimensional bottleneck codeword with a sigmoid activation at the encoder output. To adjust the model capacity while keeping the convolutional backbone unchanged, we insert an additional fully connected layer in both the encoder and decoder. By varying the dimension of this layer, we construct model variants whose total parameter counts are approximately $\times 1$, $\times 2$, and $\times 4$ relative to the GMTC dictionary size $\mathcal{O}(KM^2)$. The RD curve is generated by uniformly quantizing the bottleneck representation using different bit widths. Training is performed for 1000 epochs with an initial learning rate of $10^{-3}$, and the learning rate is decayed by a factor of $0.9$ every 100 epochs.


\begin{figure}
\centering
    \includegraphics[width=0.45\textwidth]{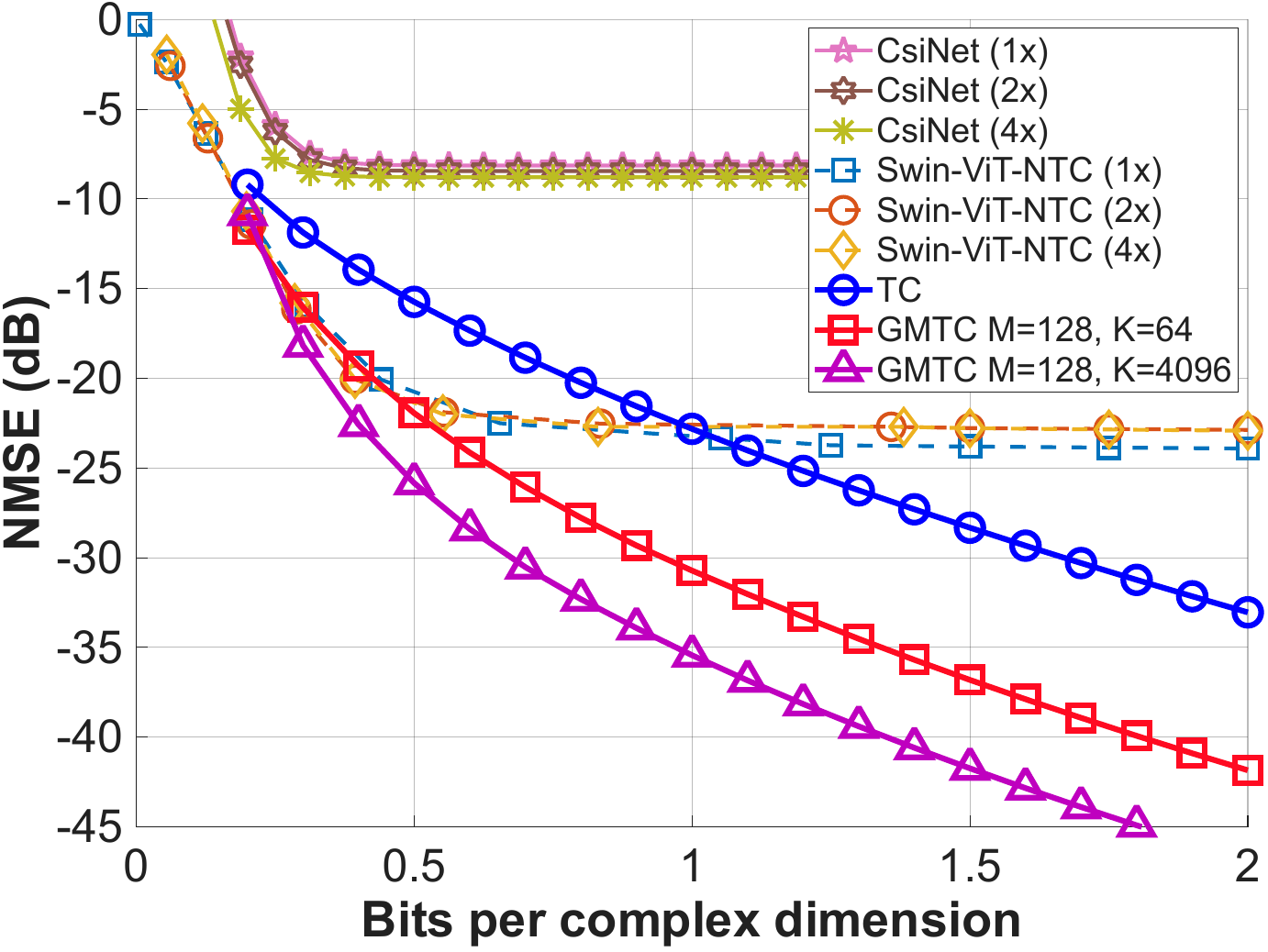}
    \caption{The RD performance validation on COST2100 data.}
    \label{fig:cost_rd_compare}
    \vspace{-4mm}
\end{figure}

{\bf Effect of $K$ and $M$:}
Figs.~\ref{fig:cost_k_sweep_n128}--\ref{fig:cost_k_sweep_n512} evaluate GMTC on COST2100 while fixing the operational blocklength ($M=128$ in Fig.~\ref{fig:cost_k_sweep_n128} and $M=512$ in Fig.~\ref{fig:cost_k_sweep_n512}) and increasing the number of mixture components $K$. Across both blocklengths, increasing $K$ yields a consistent improvement in the RD curve: for any rate, the achieved NMSE decreases monotonically as the mixture becomes richer. This trend supports the central modeling premise of the paper---massive MIMO channel is not well described by a single covariance law, but instead exhibits multiple local second-order regimes. In particular, the performance gap between TC ($K=1$) and GMTC widens as the rate increases: at very low rates, all methods are constrained to represent only coarse structure and the benefit of a more accurate covariance model is limited, whereas at moderate-to-high rates the advantage of state-adaptive, component-matched TC becomes pronounced.

The scaling behavior of these gains depends on the choice of $M$. For $M=128$, the improvement begins to saturate, showing diminishing returns as $K$ increases. In contrast, for $M=512$, further increasing $K$ yields even more significant gains, indicating that larger block lengths capture broader data distributions that necessitate a high number of mixture components for accurate modeling. Nevertheless, extreme configurations such as $M=512, K=2048$ entail prohibitive costs in memory, computation, and storage, given that the dictionary size scales as $O(KM^2)$. From a system-design perspective, $M$ and $K$ must be carefully selected to balance RD performance against hardware resources. While we employed an exhaustive grid search to validate the framework's potential, developing an efficient selection method is a critical open problem reserved for future research.  


{\bf Comparison with Baselines:} 
Fig.~\ref{fig:cost_rd_compare} compares GMTC (e.g., $M=128$, $K=64$ and $M=128$, $K=4096$) with variant baselines. The configuration with $M=128$ and $K=64$ was selected to facilitate a performance comparison under equivalent storage constraints. In this setting, the encoder consists of approximately 1 million parameters, which is comparable to the model sizes of existing neural network-based CSI compression methods\cite{liang2022changeable, park2025transformer,park2024multi, Liotopoulos2025OFSQ}. Among all pairs of M and K yielding roughly 1 million parameters, the combination of $M=128$ and $K=64$ was chosen as it demonstrated superior performance.

As shown in Fig.~\ref{fig:cost_rd_compare}, GMTC with $M=128$ and $K=64$ consistently outperforms NTC (1$\times$) over the entire operating range. While its performance is comparable to the SOTA NTC (2$\times$)  and NTC (4$\times$) at a rate range below 0.5 bits per dimension, at a rate of 1.0 bits per dimension, the NMSE gap between GMTC and baselines expands to nearly 8dB. The potential of the GMTC framework is further exemplified by the configuration with $M=128$ and $K=4096$; compared to the SOTA, it provides an NMSE reduction of roughly 5dB at 0.5 bits per dimension and exceeds a 10dB gain at 1.0 bits per dimension, demonstrating an overwhelming advantage in modeling fidelity. Notably, the performance gap between the GMTC and the neural network-based methods continues to widen as the rate grows. This highlights a fundamental advantage of the proposed framework: while neural CSI compressors--regardless of their backbone architecture or model size--typically suffer from performance saturation in high-rate regimes, GMTC achieves distortion levels that scale proportionally with the bit rate.



{\bf Complexity Comparison:}
A key operational takeaway is that GMTC's \emph{stored} model size scales as $K M^2$ (a shared dictionary of $K$ transforms). However, GMTC's \emph{online} inference cost per CSI block remains essentially the same order as classical TC: once a state is selected, the encoder/decoder applies only \emph{one} $M$-dimensional transform and performs coefficient-wise quantization/entropy coding.
Therefore, the dominant per-block cost is $O(M^2)$ for a dense transform and can be reduced to $O(M\log M)$ when a fast transform structure is used (e.g., FFT/DCT-type implementations).
In contrast, learned Swin-ViT-NTC methods require running a deep encoder/decoder network for every CSI sample, which typically incurs significantly higher compute and memory traffic at the UE.

\section{Conclusion}
\label{sec:conclusion}
 This paper studied the fundamental limits of CSI compression for FDD massive MIMO–OFDM systems under nonstationary channel statistics. Motivated by geometry-based channel models and empirical CSI datasets, we modeled the stacked wideband CSI vector as a Gaussian-mixture source, where the latent state represents different propagation regimes. This model preserves the local Gaussian structure of massive MIMO channels while capturing the multimodal statistics induced by variations in user location, scattering geometry, and blockage. Based on this model, we proposed GMTC, a practical CSI feedback architecture that combines low-rate state signaling with component-matched TC. GMTC uses a shared covariance dictionary learned offline and performs online state inference followed by TC matched to the selected covariance component. We derived converse and achievability RD bounds for Gaussian-mixture CSI. The analysis shows that the optimal distortion allocation is governed by a single global reverse-waterfilling level across all mixture components. Moreover, the gap between the converse and achievability bounds is bounded by the amortized label entropy, which vanishes as the CSI dimension grows. This provides a simple operational interpretation: the cost of nonstationary channel statistics is not a full transform description but only the number of bits required to identify the active propagation state. Simulation results using both synthetic Gaussian-mixture CSI and the COST2100 dataset confirm the theoretical insights. In particular, the mismatch penalty of single covariance TC grows as channel statistics become more multimodal, whereas GMTC closely approaches the oracle-label benchmark. Compared with representative neural CSI compressors, GMTC achieves a superior RD tradeoff while maintaining significantly lower UE-side complexity, since only a single transform operation is required once the state is identified.

Several directions remain for future work. Improving state inference under limited pilot observations may further close the gap to the oracle benchmark. Exploiting temporal correlation across CSI blocks may reduce the label overhead through predictive state tracking. Structured covariance models could further reduce the dictionary size while retaining adaptivity. Finally, integrating GMTC with channel coding, precoding, and AI-native radio architectures may enable efficient end-to-end CSI feedback solutions for future wireless systems.
\bibliographystyle{IEEEtran}
\bibliography{ref}

\end{document}